\documentclass{amsart}                    

\usepackage{times}
\usepackage{graphicx}
\graphicspath{{.}{xp/}{media/}}
\usepackage{amsmath,amsfonts,mathrsfs}
\usepackage{xspace}
\usepackage{subfigure}
\usepackage{tikz}
\usepackage{amsthm}
\newtheorem{theorem}{Theorem}
\newtheorem{lemma}{Lemma}
\usepackage[latin1]{inputenc}

\newenvironment{remark}{%
  \emph{Remark. }
}{%
}
\newenvironment{acknowledgements}{%
  \medskip
  
  \emph{Acknowledgements. }
}{%
}

\usetikzlibrary{patterns}
\usepackage{natbib}

\newcommand{\bigO}				{\text{O}}

\newcommand{\floor}[1]          {\Bigl\lfloor #1 \Bigr\rfloor}
\newcommand{\ceil}[1]           {\Bigl\lceil #1 \Bigr\rceil}

\newcommand{\nomore}[1]			{}

\newcommand{\proba}[1]           {\mathbb{P}\left\{#1\right\}}
\newcommand{\expect}[1]         {\mathbb{E}\left[ #1 \right]}
\newcommand{\var}[1]       {\text{Var}\left[ #1 \right]}
\newcommand{\sexpect}[1]         {\mathbb{E}\left[ #1 \right]}
\newcommand{\N}{\mathbb{N}}

\newcommand{\lp}{\left(}
\newcommand{\rp}{\right)}

\newcommand{\bydef}{\stackrel{\rm{def}}{=}}

\begin{document}

\title{Decentralized List Scheduling}

\author{Marc Tchiboukdjian \and Nicolas Gast \and Denis Trystram}

\address{
Marc Tchiboukdjian
CNRS / CEA,DAM,DIF - Arpajon, France}
\email{marc.tchiboukdjian@imag.fr}
\address{Nicolas Gast, EPFL, IC-LCA2, Bâtiment BC, Station 14, 1015
  Lausanne, Switzerland}
\email{nicolas.gast@epfl.ch}
\address{Denis Trystram 
Grenoble University, 51 av Jean-Kuntzman, 38330 Montbonnot, France
}
\email{denis.trystram@imag.fr}


\begin{abstract}
Classical list scheduling is a very popular and efficient technique for
scheduling jobs in parallel and distributed platforms.
It is inherently centralized. However, with the increasing number of
processors,
the cost for managing a single
centralized list becomes too prohibitive.
A suitable approach to reduce the contention is to distribute the list
among the computational units: each processor has only a local view of the work to execute.
Thus, the scheduler is no longer greedy and standard performance guarantees are lost.

\noindent
The objective of this work is to study the extra cost that must be paid
when the list is distributed among the computational units.
We first present a general methodology for computing the expected makespan
based on the analysis of an adequate potential function 
which represents the load unbalance between the local lists.
We obtain an equation on the evolution of the potential by 
computing its expected decrease in one step of the schedule.
Our main theorem shows how to solve such equations to bound
the makespan.
Then, we apply this method to 
several scheduling problems, namely, for
unit independent tasks, for weighted independent tasks and for tasks with precendence constraints.
More precisely, we prove that the time for scheduling a global workload $W$ composed of independent unit tasks
on $m$ processors is equal to $W/m$ plus an additional term proportional to $\log_2 W$.
We provide a lower bound which shows that this is optimal up to a constant.
This result is extended to the case of weighted independent tasks.
In the last setting, precedence task graphs, our analysis leads
to an improvement on the bound of \cite{ABP}.
We finally provide some experiments using a simulator. The distribution
of the makespan is shown to fit existing probability laws.
Moreover, the simulations give
a better insight on the additive term
whose value is shown to be around $3 \log_2 W$
confirming the tightness of our analysis.

\end{abstract}

\maketitle

\section{Introduction}
\label{sc:intro}

\subsection{Context and motivations}

Scheduling is a crucial issue while designing efficient parallel algorithms on new multi-core platforms.
The problem corresponds to distribute the tasks of an application (that we will called load) among available computational units
and determine at what time they will be executed.
The most common objective is to minimize 
the completion time of the latest task to be executed
(called the {\em makespan} and denoted by $C_{\max}$).
It is a hard challenging problem which received a lot of attention
during the last decade~\citep{handbook_scheduling}.
Two new books have been published recently on the topic \citep{Drozdowski,Robert-Vivien},
which confirm how active is the area.

List scheduling is one of the most popular technique for scheduling the tasks of a parallel program. 
This algorithm has been introduced by \cite{Graham69} and was used with profit in many further works
(for instance the earliest task first heuristic which
extends the analysis for communication delays in \cite{ETF}, 
for uniform machines in \cite{BenderChekuri}, or
for parallel rigid jobs in \cite{Schwiegelshohn}). 
Its principle is to build a list of ready tasks and schedule them
as soon as there exist available resources.
List scheduling algorithms are low-cost (greedy) whose performances are not too far from optimal solutions. 
Most proposed list algorithms differ in the way of considering the priority
of the tasks for building the list,
but they always consider a centralized management of the list.
However, today the parallel and distributed platforms involve more and more processors.
Thus, the time needed for managing such a centralized data structure can not be ignored anymore.
Practically, implementing such schedulers induces synchronization
overheads when several processors access the list concurrently.
Such overheads involve low-level synchronization mechanisms.

\nomore{
\subsection{Description of the problem}

We describe below the underlying computational model and define informally the problem to solve.
The parallel system is composed of a set of $m$ interconnected identical processors, that
must execute a set of tasks. The total amount of load is $W$ and can represent 
$W$ unit independent tasks, $n$ weighted tasks whose sum is equal to $W$ or tasks with dependencies.
Each processor owns its local queue of ready tasks and no one has a global view of the system. 
When the processor becomes idle, it can make a request of load to other processors in order to get some tasks.
If the request fails, another request can be sent at the next time slot.
Otherwise, a certain amount of load is transfered to this processor.
There is no explicit communication delays in this model since they are implicitly
taken into account in the cost of a load request.

In other words, the problem studied in this paper can be considered as a distributed version of the classical problems
$P \lvert \rvert C_{\max}$ and $P\lvert prec,p_j=1 \rvert C_{\max}$~\citep{handbook_scheduling}. 
}

\subsection{Related works}

Most related works dealing with scheduling consider centralized list algorithms.
However, at execution time, the cost for managing the list is neglected.
To our knowledge, the only approach that takes into account this extra management cost
is {\it work stealing}~\citep{BlumofeLeiserson} (denoted by WS in short). 

Contrary to classical centralized scheduling techniques, WS is by nature a distributed 
algorithm. Each processor manages its own list of tasks. 
When a processor becomes idle, it randomly
chooses another processor and \emph{steals} some work.
To model contention overheads, processors that request
work on the same remote list are in competition and only one can succeed.
WS has been implemented in many languages and parallel libraries including
Cilk~\citep{Cilk}, TBB~\citep{TBB} and KAAPI~\citep{Kaapi}.
It has been analyzed in a seminal paper of~\cite{BlumofeLeiserson}
where they show that the expected makespan of series-parallel
precedence graph with $W$ unit tasks on $m$ processors is bounded by
$\sexpect{C_{\max}} \le W/m + \bigO(D)$ where $D$
is the critical path of the graph (its depth).
This analysis has been improved in~\cite{ABP} using a proof
based on a potential function.
The case of varying processor
speeds has been analyzed in~\cite{BenderRabin}.
However, in all these previous analyses,
the precedence graph is constrained to have only one source and
out-degree at most $2$ which does not easily model
the basic case of independent tasks. 
Simulating independent tasks with a binary tree of precedences
gives a bound of $W/m+\bigO(\log W)$ as a complete binary tree of $W$
nodes has a depth of $D \le \log_2 W$. However, with this approach,
the structure of the binary tree dictates which tasks are stolen.
Our approach achieves a bound of the same order with a better constant and
processors are free to choose which tasks to steal. 
Notice that there exist other ways to analyze
work stealing where the work generation is probabilist
and that targets steady state results~\citep{WS_stable,WS_diff_eq,Nico}.

Another related approach which deals with distributed load balancing is 
\textit{balls into bins} games~\citep{BallsBins,WeightedBallsBins}. 
The principle is to study the maximum
load when $n$ balls are randomly thrown into $m$ bins. This is a simple distributed
algorithm which is different from the scheduling problems we are interested in. 
First, it seems hard to extend this kind of analysis for tasks with precendence constraints.
Second,
as the load balancing is done in one phase at the beginning, the cost of computing
the schedule is not considered. \cite{ParallelBallsBins} study parallel
allocations but still do not take into account contention on the bins.
Our approach, like in WS, considers contention on the lists.

Some works have been proposed for the analysis of algorithms in data structures
and combinatorial optimization
(including variants of scheduling) using potential functions. Our analysis is also
based on a potential function representing the load unbalance between the local queues.
This technique has been successfully used for analyzing 
convergence to Nash equilibria in game theory~\citep{SelfishLoad},
load diffusion on graphs~\citep{Diffusion}
and WS~\citep{ABP}. 

\subsection{Contributions}
 
List scheduling is centralized in nature.
The purpose of this work is to study the effects of decentralization on list scheduling.
The main result is a new framework
for analyzing distributed list scheduling algorithms (DLS). 
Based on the analysis of the load balancing between two processors during a work request,
it is possible to deduce the total expected number of work requests and then, to derive a bound on the expected makespan.

This methodology is generic and it is applied in this paper on several
relevant variants of the scheduling problem.
\begin{itemize}
\item We first show that the expected makespan of DLS applied on $W$ unit independent
tasks is equal to the absolute lower bound $W/m$ plus an additive term 
in $3.65\log_2 W$.
We propose a lower bound which shows that the analysis is tight up to a constant factor.
This analysis is refined and applied to several variants of the problem.
In particular, a slight change on the potential function improves the multiplicative factor from $3.65$ to $3.24$.
Then, we study the possibility of processors to cooperate while requesting some tasks in
the same list.
Finally, we study the initial repartition of the tasks and show that a balanced initial allocation induces less work requests.
\item 
Second, the previous analysis is extended to the weighted case of any unknown processing times.
The analysis achieves the same bound as before with an extra term involving $p_{\max}$ (the maximal value of the processing times).
\item Third, we provide a new analysis for the WS algorithm of \cite{ABP} for scheduling DAGs
that improves the bound on the number of work requests from 
$32mD$ to $5.5mD$.
\item Fourth, we developed a complete experimental campaign that gives statistical evidence that the makespan of DLS
follows known probability distributions depending on the considered variant. Moreover, the experiments show that the theoretical
analysis for independent tasks is almost tight: the overhead to $W/m$ is less than
$37\%$ away of the exact value.
\end{itemize}

\subsection{Content} 

We start by introducing the model and we recall the analysis for classical
list scheduling in Section~\ref{sec:model}. Then, we present the principle
of the analysis in Section~\ref{sec:thm} and
we apply this analysis on unit independent tasks in Section~\ref{sec:unit}.
Section~\ref{sec:further} discusses variations on the unit tasks model:
improvements on the potential function and cooperation among thieves.
We extend the analysis for weighted independent tasks in Section~\ref{sec:weighted}
and for tasks with dependencies in Section~\ref{sec:dag}.
Finally, we present and analyze simulation experiments in
Section~\ref{sec:simu}.

\section{Model and notations}
\label{sec:model}

\subsection{Platform and workload characteristics}

We consider a parallel platform composed of $m$ identical processors
and a workload of $n$ tasks with processing times $p_j$. The total work of
the computation is denoted by $W=\sum_{j=1}^n p_j$. The tasks can
be independent or constrained by a directed acyclic graph (DAG) of precedences.
In this case, we denote by $D$ the critical path of the DAG (its depth).
We consider an online model where the processing times and precedences are discovered
during the computation. More precisely, we learn the processing time of a task when
its execution is terminated and we discover new tasks in the DAG only when all their precedences
have been satisfied. The problem is to study the maximum completion time 
({\em makespan} denoted by~$C_{\max}$) taking into account the scheduling cost.

\subsection{Centralized list scheduling}

Let us recall briefly the principle of list scheduling as it was
introduced by \cite{Graham69}. 
The analysis states that the makespan of any list algorithm
is not greater than twice the optimal makespan.
One way of proving this bound is to use a geometric argument on the
Gantt chart: $m\cdot C_{\max} = W + S_{\text{idle}}$
where the last term is the surface of idle periods (represented
in grey in figure~\ref{fig:simu}).

Depending on the scheduling problem (with or without precedence constraints,
unit tasks or not), there are several ways to compute $S_{\text{idle}}$.
With precedence constraints, $S_{\text{idle}} \le (m-1) \cdot D$.
For independent tasks, the results can be written as
$S_{\text{idle}} \le (m-1) \cdot p_{\max}$ where $p_{\max}$ is the maximum of the
processing times.
For unit independent tasks, it is straightforward to obtain an optimal algorithm 
where the load is evenly balanced. Thus $S_{\text{idle}} \le m-1$, \textit{i.e.}
at most one slot of the schedule contains idle times.

\subsection{Decentralized list scheduling}

When the list of ready tasks is distributed among the processors,
the analysis is more complex even in the elementary case of unit independent
tasks. In this case, the extra $S_{\text{idle}}$ term is induced
by the distributed nature of the problem. Processors can be idle
even when ready tasks are available.
Fig.~\ref{fig:simu} is an example of a schedule obtained using distributed list
scheduling which shows the complicated repartition of the idle times $S_{\text{idle}}$.

\begin{figure}[tb]
\centering
\includegraphics[width=\columnwidth]{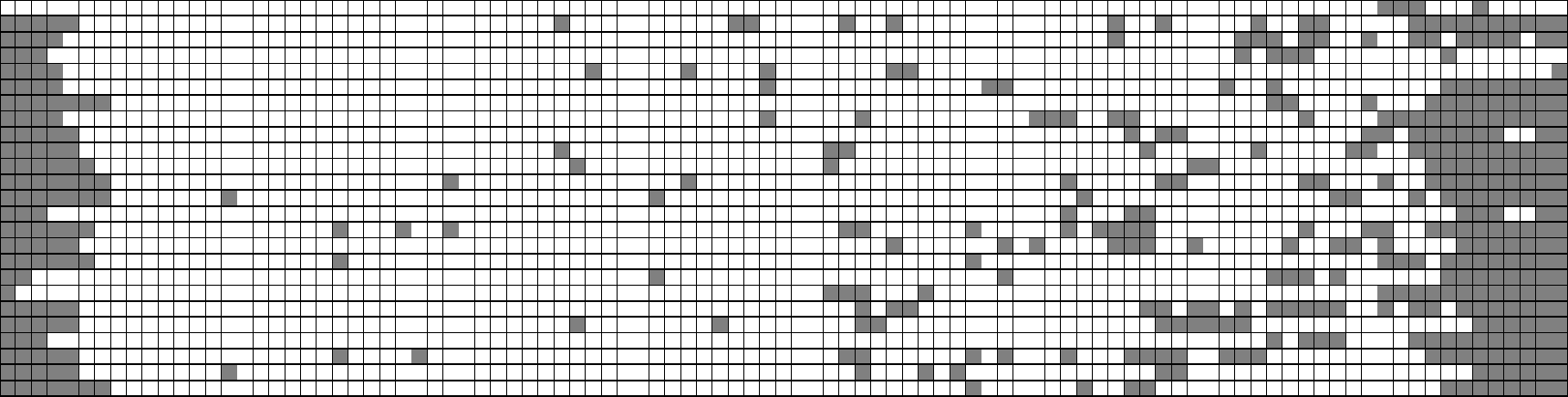}
\caption{A typical execution of $W=2000$ unit independent tasks on $m=25$ processors
using distributed list scheduling. Grey area represents idle times due to steal requests.
}
\label{fig:simu}
\end{figure}

\subsection{Model of the distributed list}

We now describe precisely the behavior of the distributed list.
Each processor $i$ maintains its own local queue $Q_i$ of tasks ready to execute.
At the beginning of the execution, ready tasks can be arbitrarily spread among the queues.
While $Q_i$ is not empty, processor $i$ picks a task and executes it. 
When this task has been executed, it is removed from the queue
and another one starts being processed.
When $Q_i$ is empty, processor $i$ sends a \emph{steal request}
to another processor $k$ chosen uniformly at random.
If $Q_k$ is empty or contains only one task (currently executed
by processor $k$), then the request fails and processor $i$ will send a new request
at the next time step. If $Q_k$ contains more than one task, then $i$ is given half
of the tasks and it will restart a normal execution at the next step. 
To model the contention on the queues, no more than one steal request per processor can succeed
in the same time slot. If several requests target the same processor, a random one
succeeds and all the others fail. This assumption will be relaxed in Section~\ref{ssec:coop}.
A steal request is said \textit{successful} if the target queue contains more than one task
and the request is not aborted due to contention. In all the other cases, the steal request
is said \textit{unsuccessful}.

This is a high level model of a distributed list but it accurately models 
the case of independent tasks and the WS algorithm of~\cite{ABP}.
We justify here some choices of this model.
There is no explicit communication cost since WS algorithms most often target
shared memory platforms.
In addition, a steal request is done in constant time independently of the
amount of tasks transfered.
This assumption is not restrictive as
the description of a large number of tasks can be very short. 
In the case of independent tasks, a whole subpart of an array of tasks
can be represented in a compact way by the range of the corresponding indices,
each cell containing the effective description of a task
(a STL transform in~\cite{JL_STL}).
For more general cases with precedence constraints,
it is usually enough to transfer a task which represents a part of the DAG. 
More details on the DAG model are provided in Section~\ref{sec:dag}.
Finally, there is no contention between a processor executing a task from its own
queue and a processor stealing in the same queue. Indeed, one can use 
queue data structures allowing these two operations to happen concurrently~\citep{Cilk}.

\subsection{Properties of the work}

At time $t$, let $w_i(t)$ represent the amount of work in queue $Q_i$ (\textit{cf.} Fig.~\ref{fig:workload_evolution}).
$w_i(t)$ may be defined as the sum of processing times of all tasks in $Q_i$
as in Section~\ref{sec:unit} but can differ as in Sections~\ref{sec:weighted} and \ref{sec:dag}.
In all cases, the definition of $w_i(t)$ satisfies the following properties.
\begin{enumerate}
\item When $w_i(t)>0$, processor $i$ is active and executes some work: $w_i(t+1)\leq w_i(t)$.
\item When $w_i(t)=0$, processor $i$ is idle and send a steal request to a random processor $k$.
If the steal request is successful, a certain amount of work is transfered from
processor $k$ to processor $i$ and we have $\max\{w_i(t+1),w_k(t+1)\} <  w_k(t)$.
\item The execution terminates when there is no more work in the system, \textit{i.e.} $\forall i, w_i(t)=0$.
\end{enumerate}

We also denote the total amount of work on all processors by $w(t) = \sum_{i=1}^m w_i(t)$
and the number of processors sending steal requests by $r_t \in [0,m-1]$. 
Notice that when $r_t=m$, all queues are empty and thus the execution is complete.


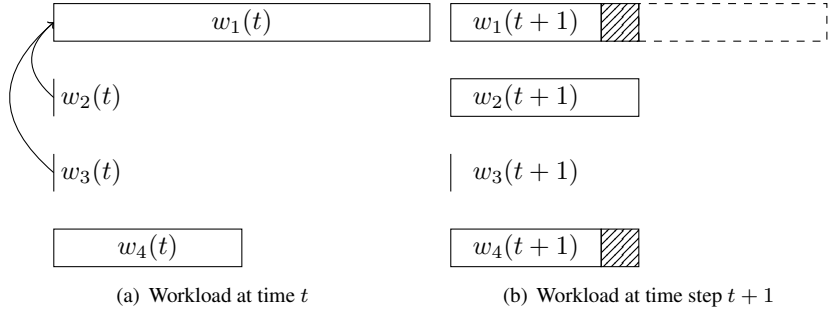
\begin{figure}[ht]
  \centering
  \subfigure[Workload at time $t$]{
    \begin{tikzpicture}[yscale=.5]
      \draw (0,0) rectangle (2.5,1); \node at (1.25,.5) {$w_4(t)$};
      \draw (0,2) rectangle (0,3); \node at (0.5,2.5) {$w_3(t)$};
      \draw (0,4) rectangle (0,5); \node at (0.5,4.5) {$w_2(t)$};
      \draw (0,6) rectangle (5,7); \node at (2.5,6.5) {$w_1(t)$};
      \draw (0,2.5) edge[bend left,->] (0,6.5);
      \draw (0,4.5) edge[bend left,->] (0,6.5);
    \end{tikzpicture}
  }~~~~~~~
  \subfigure[Workload at time step $t+1$]{
    \begin{tikzpicture}[yscale=.5]
      \draw (0,0) rectangle (2.,1);  \node at (1, .5) {$w_4(t+1)$};
      \draw (0,2) rectangle (0,3);   \node at (1,2.5) {$w_3(t+1)$};
      \draw (0,4) rectangle (2.5,5); \node at (1,4.5) {$w_2(t+1)$};
      \draw (0,6) rectangle (2.,7);  \node at (1,6.5) {$w_1(t+1)$};
      \tikzstyle{stollen}=[dashed];
      \draw[stollen] (2.5,6) rectangle (5,7); 
      \tikzstyle{executed}=[pattern=north east lines];
      \draw[executed] (2.,6) rectangle (2.5,7); 
      \draw[executed] (2.,0) rectangle (2.5,1); 
    \end{tikzpicture}
  }
  
  \caption{Evolution of the workload of the different processors during a
    time step. At time $t$, processors $2$ and $3$ are idle and they both
    choose processor $1$ to steal from. At time $t+1$, only processor $2$
    succeed in stealing some of the work of processor $1$. The work is
    split between the two processors. Processors $1$ and $4$ both execute
    some work during this time step (represented by a shaded zone). }
  \label{fig:workload_evolution}
\end{figure}

\section{Principle of the analysis and main theorem}
\label{sec:thm}

\newcommand\F{\mathscr{F}}
\newcommand\w{\mathbf{w}}
\newcommand\R{\mathbb{R}}

This section presents the principle of the analysis. The main result is
Theorem~\ref{thm:main} that gives bounds on the expectation of the
steal requests done by the schedule as well as the probability that the
number of work requests exceeds this bound. As a processor is either
executing or requesting work, the number of work requests plus the total
amount of tasks to be executed is equal to $m\cdot C_{\max}$, where
$C_{\max}$ is the total completion time. The makespan can be derived from
the total number of work requests:
\begin{equation}
  C_{\max}=\frac{W}{m}+\frac{R}{m}.
  \label{eq:Cmax=w+r}
\end{equation} 

The main idea of our analysis is to study the decrease of a potential
$\Phi_t$.
The potential $\Phi_t$ depends on the load on all processors at time $t$,
$\w(t)$. The precise definition of $\Phi_t$ varies depending on the
scenario (see Sections~\ref{sec:unit} to \ref{sec:dag}). For example, the
potential function used in Section~\ref{sec:unit} is
$\Phi_t=\sum_{i=1}^m (w_i(t)-w(t)/m)^2$.
For each scenario, we will prove that the diminution of the potential
during one time step depends on the number of steal requests, $r_t$. More
precisely, we will show that there exists a function $h:\{0\dots
m\}\to[0;1]$ such that the average value of the potential at time $t+1$
is less than $\Phi_t/h(r_t)$.

Using the expected diminution of the potential, we derive a bound on the
number of steal requests until $\Phi_t$ becomes less than one,
$R=\sum_{s=0}^{\tau-1} r_s$, where $\tau$ denotes the first time that
$\Phi_t$ is less than $1$. If all $r_t$ were equal to $r$ and the potential
decrease was deterministic, the number of time steps before $\Phi_t\le 1$
would be $\ceil{\log\Phi_0/\log h(r)}$ and the number of steal requests 
would be $r/\log h(r) \log\Phi_0$. As $r$ can vary between $1$ and $m$, the
worst case for this bound is $m\lambda\cdot\log\Phi_0$, where
$m\lambda=\max_{1\le r\le m} r/\log(h(r))$. 

The next theorem shows that number of steal requests is indeed bounded by
$m\lambda\log\Phi_0$ plus an additive term due to the stochastic nature
of $\Phi_t$. The fact that $\lambda$ corresponds to the worst choice of $r_t$
at each time step makes the bound looser than the real constant. However,
we show in Section~\ref{sec:simu} that the gap between the obtained bound
and the values obtained by simulation is small.  Moreover, the computation
of the constant $\lambda$ is simple and makes this analysis applicable
in several scenarios, such as the ones presented in Sections \ref{sec:unit}
to \ref{sec:dag}.

\medskip
In the following theorem and its proof, we use the following notations.
$\F_t$ denotes the knowledge of the system up to time $t$ (namely, the
filtration associated to the process $\w(t)$). For a random variable $X$,
the conditional expectation of $A$ knowing $\F_t$ is denoted
$\expect{X\mid\F_t}$. Finally, the notation $\mathbf{1}_A$ denotes the
random variable equal to $1$ if the event $A$ is true and $0$ otherwise. In
particular, this means that the probability of an event $A$ is
$\proba{A}=\expect{\mathbf{1}_A}$.
\begin{theorem} \label{thm:main}
  Assume that there exists a function $h:\{0\dots m\}\to[0,1]$ such that
  the potential satisfies:  
  $$\sexpect{\Phi_{t+1}\mid\F_t}\le h(r_t) \cdot \Phi_t.$$ 
  Let $\Phi_0$ denotes the potential at time $0$ and $\lambda$ be defined
  as:
  \begin{equation*}
    \lambda\bydef\max_{1\le r\le m} \frac{r}{-m\log_2(h(r))}
  \end{equation*}
  Let $\tau$ be the first time that $\Phi_t$ is less than $1$,
  $\tau\bydef\min\{t:\Phi_t<1\}$.  The number of steal requests until
  $\tau$, $R=\sum_{s=0}^{\tau-1} r_s$,  satisfies:
  \begin{itemize}
  \item[(i)] 
    $\displaystyle \proba{R \ge m \cdot \lambda \cdot \log_2 \Phi(0) + m +
      u} \le 2^{-u/(m\cdot\lambda)}$
  \item[(ii)]$\displaystyle \sexpect{R} \le m\cdot\lambda \cdot \log_2
    \Phi(0) + m(1+\frac{\lambda}{\ln2}).$
  \end{itemize}
\end{theorem}

\begin{proof}
  For two time steps $t\le
  T$, we call $R^T_t$ the number of steal requests between $t$ and $T$:
  \begin{equation*}
    R^T_t \bydef \sum_{s=t}^{\min \{\tau,T\}-1}r_s.
  \end{equation*}
  The number of steal requests until $\Phi_t<1$ is $R =
  \sum_{s=0}^{\tau-1}r_s = \lim_{T\to\infty} R^T_0$.
  
  We show by a backward induction on $t$ that for all $t\le T$:
  \begin{equation}
    \mathrm{if~}\Phi_t\ge1, \mathrm{then~}\forall u\in\R:
    \expect{\mathbf{1}_{R_t^T\ge
        m\cdot\lambda\cdot\log_2\Phi_t+m+u}\mid\F_t} 
    \le2^{-u/(m\cdot\lambda)}. 
    \label{eq:HR}
  \end{equation}
  
  For $t{=}T$, $R^T_T=0$ and $\expect{\mathbf{1}_{R_t^T\ge
      m\cdot\lambda\cdot\log_2\Phi_t+m+u}\mid\F_t}=0$. Thus, \eqref{eq:HR}
  is true for $t{=}T$.
  
  Assume that \eqref{eq:HR} holds for some $t+1\le T$ and suppose that
  $\Phi_t\ge1$. Let $u>0$ (if $u\le0$ ... ). Since $R_t^T=r_t+R_{t+1}^T$, the probability
  $\proba{R_t^T\ge m\cdot\lambda\cdot\log_2\Phi_t+m+u\mid\F_t}$ is equal to
  \begin{align}
    \expect{\mathbf{1}_{R_{t}^T\ge
        m\lambda\log_2\Phi_t+m+u}\mid\F_t}&=
    \expect{\mathbf{1}_{r_t+R_{t+1}^T\ge
        m\lambda\log_2\Phi_t+m+u}\mid\F_t}\\
    &=\expect{\mathbf{1}_{r_t+R_{t+1}^T\ge
        m\lambda\log_2\Phi_t+m+u}\mathbf{1}_{\Phi_{t+1}\ge1}\mid\F_t}
    \label{eq:probat+1}\\
    &~~+ \expect{\mathbf{1}_{r_t+R_{t+1}^T\ge 
        m\lambda\log_2\Phi_t+m+u}\mathbf{1}_{\Phi_{t+1}<1}\mid\F_t} 
    \label{eq:zero}
  \end{align}
  If $\Phi_{t+1}<1$, then $R_{t+1}^T=0$. Since $m\ge r_t$ and $\Phi_t\ge1$, 
  $m\lambda\log_2\Phi_t+m+u-r_t\ge0$. This shows that the term of
  Equation~\eqref{eq:zero} is equal to zero. \eqref{eq:probat+1} is the
  probability that $R^T_{t+1}$ is greater than
  \begin{equation*}
    m\lambda\log_2\Phi_t+m+u-r_t=m\lambda\log_2\Phi_{t+1}+m+(u-r_t-m\lambda\log(\Phi_{t+1}/\Phi_{t})
  \end{equation*}
  Therefore, using the induction hypothesis, \eqref{eq:probat+1} is equal 
  to 
  \begin{align*}
    \expect{\mathbf{1}_{R_{t+1}^T\ge
        m\lambda\log_2\Phi_t+m+u-r_t}\mathbf{1}_{\Phi_{t+1}>1}\mid\F_t}&=
    \expect{2^{-\frac{u-r_t-m\lambda\log(\Phi_{t+1}/\Phi_t)}{m\lambda}}\mathbf{1}_{\Phi_{t+1}>1}\mid\F_t}\\  
    &=2^{-\frac{u-r_t}{m\lambda}}\expect{\frac{\Phi_{t+1}}{\Phi_t}\mathbf{1}_{\Phi_{t+1}>1}\mid\F_t}\\
    &=2^{-\frac{u-r_t}{m\lambda}}h(r_t)\\
    &=2^{-\frac{u}{m\lambda}}2^{r_t/\lambda+\log_2(h(r_t))},
  \end{align*}
  where at the first line we used both the fact that for a random variable
  $X$, $\expect{X\mid\F_t}=\expect{\expect{X\mid\F_{t+1}}\mid\F_t}$ and the
  induction hypothesis. 
  
  If $r_t=0$, $2^{r_t/\lambda+\log_2(h(r_t))}=h(r_t)\le1$. Otherwise, by
  definition of $\lambda=\max_{1r\le m}r/-\log(h(r))$,
  $r_t/\lambda+\log_2(h(r_t))\le0$ and
  $2^{r_t/\lambda+\log_2(h(r_t))}\le1$. This shows that \eqref{eq:HR} holds
  for $t$. Therefore, by induction on $t$, this shows that \eqref{eq:HR}
  holds for $t=0$: for all $u\ge0$:
  \begin{equation*}
    \expect{\mathbf{1}_{R_0^T\ge
        m\cdot\lambda\cdot\log_2\Phi_t+m+u}\mid\F_0}
    \le2^{-u/(m\cdot\lambda)}
  \end{equation*}
  As $r_t\ge0$, the sequence $(R^T_0)_T$ is increasing and converges to
  $R$. Therefore, the sequence $\mathbf{1}_{R_0^T\ge
    m\cdot\lambda\cdot\log_2\Phi_0+m+u}$ is increasing in $T$ and converges
  to $\mathbf{1}_{R\ge m\cdot\lambda\cdot\log_2\Phi_0+m+u}$.  Thus, by
  Lebesgue's monotone convergence theorem, this shows that
  \begin{equation*}
    \proba{R\ge
      m\cdot\lambda\cdot\log_2\Phi_0+m+u}=\lim_{T\to\infty}
    \expect{\mathbf{1}_{R_0^T\ge m\cdot\lambda\cdot\log_2\Phi_0+m+u}}
    \le 2^{-\frac{u}{m\lambda}}.
  \end{equation*}
  
  The second part of the theorem \emph{(ii)} is a direct consequence of
  \emph{(i)}. Indeed,
  \begin{align*}
    \expect{R} &=\int_0^\infty \proba{R\ge u}du\\
    &\le m\cdot\lambda\cdot\log_2\Phi_0+m + \int_{0}^\infty 
    \proba{R\ge m\cdot\lambda\cdot\log_2\Phi_0+m+u} du \\ 
    &\le m\cdot\lambda\cdot\log_2\Phi_0+m +
    \int_{0}^\infty2^{-\frac{u}{m\lambda}}du\\
    &\le m\cdot\lambda\cdot\log_2\Phi_0+m(1+\frac{\lambda}{\ln 2}).
  \end{align*}
  \qed
\end{proof}

\section{Unit independent tasks}
\label{sec:unit}

We apply the analysis presented in the previous section for the case of
independent unit tasks.  In this case, each processor $i$ maintains a local
queue $Q_i$ of tasks to execute.  At every time slot, if the local queue
$Q_i$ is not empty, processor $i$ picks a task and executes it. When $Q_i$
is empty, processor $i$ sends a steal request to a random processor $j$. If
$Q_j$ is empty or contains only one task (currently executed by processor
$j$), then the request fails and processor $i$ will have to send a new
request at the next slot. If $Q_j$ contains more than one task, then $i$ is
given half of the tasks (after that the task executed at time $t$ by
processor $j$ has been removed from $Q_j$).  The amount of work on
processor $i$ at time $t$, $w_i(t)$, is the number of tasks in $Q_i(t)$.
At the beginning of the execution, $w(0)=W$ and tasks can be arbitrarily
spread among the queues.

\subsection{Potential function and expected decrease}

Applying the method presented in Section~\ref{sec:thm}, the first step of
the analysis is to define the potential function and compute the potential
decrease when a steal occurs. For this example, $\Phi(t)$ is defined by:
\begin{equation*}
  \Phi(t) = \sum_{i=1}^{m} \lp w_i(t)-\frac{w(t)}{m}\rp^2 = 
  \sum_{i=1}^{m} w_i(t)^2-\frac{w^2(t)}{m}.
\end{equation*}
This potential represents the load unbalance in the system. If all queues
have the same load $w_i(t)=w(t)/m$, then $\Phi(t)=0$.  $\Phi(t)\le 1$ implies
that there is at most one processor with at most one more task than the
others. In that case, there will be no steal until there is just one
processor with $1$ task and all others idle. Moreover, the potential
function is maximal when all the work is concentrated on a single
queue. That is $\Phi(t) \le w(t)^2 -w(t)^2/m \le (1-1/m)w^2(t)$.

Three events contribute to a variation of potential: successful steals,
tasks execution and decrease of $w^2(t)/m$.
\begin{enumerate}
\item If the queue $i$ has $w_i(t)\ge1$ tasks and it receives one or more
  steal requests, it chooses a processor $j$ among the thieves. At time
  $t+1$, $i$ has executed one task and the rest of the work is split
  between $i$ and $j$. Therefore,
  \begin{equation*}
    w_i(t+1)=\ceil{(w_i(t)-1)/2} \quad\mathrm{and}\quad
    w_j(t+1)=\floor{(w_i(t)-1)/2}.
  \end{equation*}
  Thus, we have:
  \begin{equation*}
    w_i(t+1)^2+w_j(t+1)^2 =
    \ceil{(w_i(t){-}1)/2}^2+\floor{(w_i(t){-}1)/2}^2 \le w_i(t)^2/2-w_i(t)+1. 
  \end{equation*}
  Therefore, this generates a difference of potential of
  \begin{equation}
    \delta_i(t) \ge w_i(t)^2/2+w_i(t)-1.
    \label{eq:delta_i^k}
  \end{equation}
\item If $i$ has $w_i(t)\ge1$ tasks and receives zero steal requests, it
  potential goes from $w_i(t)^2$ to $(w_i(t)-1)^2$, generating a potential
  decrease of $2w_i(t)-1$.
\item As there are $m-r_t$ active processors, $(\sum_{i=1}^mw_i(t))^2/m$ goes
  from $w(t)^2/m$ to $w(t+1)^2=(w(t)-m+r)^2/m$, generating a potential
  increase of $2(m-r_t)w(t)/m-(m-r_t)^2/m$. 
\end{enumerate}
Recall that at time $t$, there are $r_t$ processors that send steal
requests. A processor $i$ receives zero steal requests if the $r_t$ thieves
choose another processor. Each of these events is independent and happens
with probability $(m-2)/(m-1)$. Therefore, the probability for the
processor to receive one or more steal requests is $q(r_t)$ where
\begin{equation*}
  q(r_t) = 1 - \lp 1-\frac{1}{m-1}\rp^{r_t}. 
\end{equation*}
If $\Phi_t{=}\Phi$ and $r_t{=}r$, by summing the expected decrease on each
active processor $\delta_i$, the expected potential decrease is greater
than:
\begin{eqnarray}
  \sum_{i/w_i(t){>}0} \left[ q(r) \underbrace{\Bigl(
    \frac{w_i(t)^2}{2}+w_i(t){-}1\Bigr)}_{\mathrm{\ge\delta_i}}+(1-q(r))(2w_i(t){-}1)
\right]
  - 2w(t)\frac{m-r}{m}+\frac{(m-r)^2}{m} \nonumber\\
  =\left[\sum_{i/w_i(t)>0}\frac{q(r)}{2}w_i(t)^2\right] - q(r)w(t) +2w(t) - (m-r)
  - 2w(t)\frac{m-r}{m}+\frac{(m-r)^2}{m}. \nonumber
\end{eqnarray}
Using that $2w(t)-2w(t)\frac{m-r}{m}=2w(t)\frac{r}{m}$, that
$-(m-r)+\frac{(m-r)^2}{m}=-(m-r)\frac{r}{m}$ and that $\sum w_i(t)^2 =
\Phi+w(t)^2$, this equals:
\begin{align*}
  &\frac{q(r)}{2}\Phi + \frac{q(r)}{2}\frac{w(t)^2}{m}
  -q(r)w(t)+2w(t)\frac{r}{m}-(m-r)\frac{r}m\\ 
  &\quad=\frac{q(r)}{2}\Phi+\frac{q(r)}{2}\frac{w(t)^2}{m} -q(r)w(t) +
  \frac{r}{m}(2w(t)-m+r)\\ 
  &\quad=\frac{q(r)}{2}\Phi+\frac{q(r)w(t)}{2}\lp\frac{w(t)}{m}-2+\frac{2r}{mq(r)}\rp
  +\frac{r}{m}\lp w(t)-m+r\rp.  
\end{align*}
By concavity of $x\mapsto(1-(1-x)^r)$, $(1-(1-x)^r)\le r\cdot x$. This
shows that $q(r)=1-(1-\frac1{m-1})^r\le 
r/(m-1)$. Thus, $r/q(r)\ge m-1$. Moreover, as $m-r$ is the number of active
processors, $w\ge m-r$ (each processor has at least one task). This shows
that the expected decrease of potential is greater than:
\begin{align*}
  \frac{q(r)}{2}\Phi+\frac{q(r)w(t)}{2}\lp \frac{w(t)}m - 2
  +2\frac{m-1}{m}\rp
  =\frac{q(r)}{2}\Phi+\frac{q(r)w(t)}{2m}(w(t)-2). 
\end{align*}
If $w(t)\ge2$, then the expected decrease of potential is greater than
$q(r_t)\Phi_t/2$. If $w(t)<2$, this means that $w(t)=1$ and $w(t+1)=0$ and
therefore $\Phi_{t+1}=0$. Thus, for all $t$:
\begin{equation}
  \expect{\Phi_{t+1}\mid\F_t}\le \Bigl(1-\frac{q(r_t)}{2}\Bigr)\cdot \Phi_t.
  \label{eq:phi_t+1_unit}
\end{equation}

\subsection{Bound on the makespan}

Using Theorem~\ref{thm:main} of the previous section, we can solve
equation \eqref{eq:phi_t+1_unit} and conclude the analysis.
\begin{theorem}
  \label{thm:unit}
  Let $C_{\max}$ be the makespan of $W=n$ unit independent tasks scheduled by
  DLS and $\Phi_0\bydef\sum_{i}(w_i-\frac{W}{m})^2$ the potential when the
  schedule starts.  Then:
  \begin{itemize}
  \vspace{-1ex}
  \item[(i)]
    ~~~~~$\displaystyle\expect{C_{\max}} \le \frac{W}{m} +
    \frac{1}{1-\log_2(1+\frac{1}{e})}\cdot\Bigl(\log_2 \Phi_0+\frac{1}{\ln2}\Bigr) + 1 $ 
  \item[(ii)]
    ~~~~~$\displaystyle\proba{C_{\max} \ge \frac{W}{m} +
      \frac{1}{1-\log_2(1+\frac{1}{e})}\cdot\lp\log_2
      \Phi_0+\log_2\frac{1}\epsilon\rp + 1} \le \epsilon$
  \end{itemize}
  In particular:
  \begin{itemize}
  \item[(iii)] 
    ~~~~~$\displaystyle\expect{C_{\max}} \le \frac{W}{m} +
    \frac{2}{1-\log_2(1+\frac{1}{e})}\cdot\Bigl(\log_2 W+\frac{1}{2\ln2}\Bigr) + 1$
  \end{itemize}
  These bounds are optimal up to a constant factor in $\log_2 W$.
\end{theorem}

\begin{proof}
  Equation~\eqref{eq:phi_t+1_unit} shows that $\sexpect{\Phi_{t+1}|\F_t}\le
  g(r_t)\Phi_t$ with $g(r)=1-q(r)/2$. Defining
  $\Phi'_t=\Phi_t/(1-1/(m-1))$, the potential function $\Phi'_t$ also
  satisfies \eqref{eq:phi_t+1_unit}. Therefore, $\Phi'_t$ satisfies the
  conditions of Theorem~\ref{thm:main}. This shows that the number of work
  requests $R$ until $\Phi'_t<1$ satisfies 
  \begin{equation*}
    \sexpect{R} \le m\cdot\lambda\log_2(\Phi_0)+m\Bigl(1+\frac{\lambda}{\ln2}\Bigr),
  \end{equation*}
  with $\lambda=\max_{1\le r\le m-1}
  r/(-m\log_2h(r))$.
One can show
  that $r/(-m\log_2h(r))$ is decreasing in $r$.
  Thus its minimum is attained for
  $r=1$. This shows that $\lambda\le 1/(1-\log_2(1+\frac{1}{e}))$.
  
  The minimal non zero-value for $\Phi_t$ is when one processor has one
  task and the others zero. In that case, $\Phi_t=1-1/(m-1)$. Therefore,
  when $\Phi'_t<1$, this means that $\Phi_t=0$ and the schedule is
  finished. 
  
  As pointed out in Equation~\eqref{eq:Cmax=w+r}, at each time step of the
  schedule, a processor is either computing one task or stealing
  work. Thus, the number of steal requests plus the number of tasks to be
  executed is equal to $m\cdot C_{\max}$, \emph{i.e.}  $m\cdot C_{\max} =
  W+R$. This shows that
  \begin{equation*}
    \displaystyle \sexpect{C_{\max}} \le \frac{W}{m} +
    \frac{1}{1-\log_2(1+\frac{1}{e})}\cdot\lp\log_2 \Phi_0+\frac{1}{\ln2}\rp + 1. 
  \end{equation*}
  
  This concludes the proof of \emph{(i)}. 
  The proof of the \emph{(i)} applies \textit{mutatis mutandis} to prove
  the bound in probability \emph{(ii)} using
  Theorem~\ref{thm:main}~\emph{(ii)}.

  We now give a lower bound for this problem. Consider $W=2^{k+1}$ tasks
  and $m=2^k$ processors, all the tasks being on the same processor at the
  beginning.  In the best case, all steal requests target processors with
  highest loads. In this case the makespan is $C_{\max} = k+2$: the first
  $k=\log_2 m$ steps for each processor to get some work; one step where
  all processors are active; and one last step where only one processor is
  active.  In that case, $C_{\max} \ge \frac{W}{m} + \log_2 W - 1$.
  \qed
\end{proof}

This theorem shows that the factor before $\log_2 W$ is bounded by $1$ and
$2/(1-\log_2(1+1/e))< 3.65$. Simulations reported in Section~\ref{sec:simu}
seem to indicate that the factor of $\log_2 W$ is slightly less than $3.65$.
This shows that the constants obtained by our analysis are sharp.

\subsection{Influence of the initial repartition of tasks}

In the worst case, all tasks are in the same queue at the beginning of the
execution and $\Phi_0=(W-W/m)^2\le W^2$. This leads to a bound 
on the number of work requests in $3.65 m \log_2 W$
(see the item \emph{(iii)} of Theorem~\ref{thm:unit}). However, using bounds
in terms of $\Phi_0$, our
analysis is able to capture the difference for the number of work requests
if the initial repartition is more balanced.  One can show that a
more balanced initial repartition ($\Phi_0 \ll W^2$)
leads to fewer steal requests on average.

Suppose for example that the initial repartition is a balls-and-bins
assignment: each tasks is assigned to a processor at random. 
In this case, the initial number of tasks in queue $i$, $w_i(0)$, follows
a binomial distribution $\mathcal{B}(W,1/m)$.
The expected value of $\Phi_0$ is:
\begin{equation*}
  \displaystyle \expect{\Phi_0} = \sum_i \expect{w_i^2} - \frac{W^2}{m} 
  = \sum_i \Bigl(\var{w_i}+\expect{w_i}^2\Bigr) -\frac{W^2}{m}
  = \Bigl(1-\frac1m\Bigr)\cdot W
\end{equation*}
Since the number of work requests
is proportional to $\log_2 \Phi_0$, this initial repartition of tasks
reduces the number of steal requests by a factor of $2$ on average.
This leads to a better bound on the makespan in $W/m+1.83\log_2W+3.63$.

\section{Going further on the unit tasks model}
\label{sec:further}

In this section, we provide two different analysis of the model of unit
tasks of the previous section.  We first show how the use of a different
potential function $\Phi_t=\sum_iw_i(t)^\nu$ (for some $\nu>1$) leads to a
better bound on the number of work requests. Then we show how
cooperation among thieves leads to a reduction of the bound on the number
of work requests by $12\%$. The later is corroborated by our simulation
that shows a decrease on the number of work requests between $10\%$ and
$15\%$.

\subsection{Improving the analysis by changing the potential function}
\label{ssec:new-potential}

We consider the same model of unitary tasks as in
Section~\ref{sec:unit}. The potential function of our system is defined as
\begin{equation*}
  \Phi_t = \sum_{i=1}^mw_i(t)^\nu,
\end{equation*}
where $\nu>1$ is a constant factor. 

When an idle processor steals a processor with $w_i(t)$ tasks, the
potential decreases by
\begin{align*}
  \delta_i = w_i(t)^\nu-\ceil{\frac{w_i(t)-1}{2}}^\nu +
  \floor{\frac{w_i(t)-1}{2}}^\nu
  &\ge w_i(t)^\nu-\floor{\frac{w_i(t)}{2}}^\nu +
  \floor{\frac{w_i(t)}{2}}^\nu \\
  &\ge\lp1-2^{1-\nu}\rp w_i(t)^\nu.  
\end{align*}
This shows that the expected value of the potential at time $t+1$ is 
\begin{equation*}
  \expect{\Phi_{t+1}}\le(1-q(r)(1-2^{1-\nu})) \cdot \Phi_t. 
\end{equation*}
where $q(r)$ is the probability for a processor to receive at least one work
request when $r$ processors are stealing,
$q(r)=1-\lp1-\frac{1}{m-1}\rp^{r}$.

Following the analysis of the previous part, and as $\Phi_0\le W^\nu$ the
expected makespan is bounded by:
\begin{equation*}
  \frac{W}{m}+\lambda(\nu)\cdot\lp\log\Phi_0+1+\frac1{\ln2}\rp \le 
  \frac{W}{m}+\nu\lambda(\nu)\cdot\lp\log W+1+\frac1{\ln2}\rp,
\end{equation*}
where $\lambda(\nu)$ is a constant depending on $\nu$ equal to:
\begin{equation}
  \lambda(\nu)\bydef\max_r\Bigl\{\frac{r}{-\log_2(1-q(r)(1-2^{1-\nu}))}\Bigr\}
  \label{eq:max_nu}
\end{equation}
As for $\nu=2$ of Section~\ref{sec:unit}, it can be shown the maximum of
Equation~\ref{eq:max_nu} is attained for $r=m-1$. 

The constant factor in front of $\log W$ is $\nu\lambda(\nu)$.
Numerically, the minimum of $\nu\lambda(\nu)$ is for $\nu\approx2.94$ and
is less than $3.24$.
\begin{theorem}
  \label{thm:unit_nu}
  Let $C_{\max}$ be the makespan of $W=n$ unit independent tasks scheduled DLS.
Then:
  \begin{equation*}
    \displaystyle\expect{C_{\max}} \le \frac{W}{m} +
    3.24\cdot\Bigl(\log_2 W+\frac{1}{2\ln2}\Bigr) + 1
  \end{equation*}
\end{theorem}
In Section~\ref{sec:unit}, we have shown that the makespan was bounded by
\begin{equation*}
  \frac{W}{m} +2\lambda(2)\cdot\Bigl(\log_2 \Phi_0+\frac{1}{\ln2}\Bigr) + 1 \le
    \frac{W}{m}+3.65\cdot\Bigl(\log_2W+\frac{1}{2\ln2}\Bigr)+1.
\end{equation*}
Theorem~\ref{thm:unit_nu} improves the constant factor in front of
$\log_2 W$.
However, we loose the information of the initial repartition of tasks
$\Phi_0$.

\subsection{Cooperation among thieves}
\label{ssec:coop}

In this section, we modify the protocol for managing the distributed list.
Previously, when $k>1$ steal requests were sent on the same processor, only
one of them could be served due to contention on the list. We now
allow the $k$ requests to be served in unit time.
This model has been implemented in the middleware Kaapi~\citep{Kaapi}.
When $k$ steal requests target the same processor, the work is divided into $k+1$
pieces. In practice, allowing concurrent thieves increase the cost of a steal
request but we neglect this additional cost here. 
We assume that the $k$ concurrent steal requests can be served in 
unit time. We study the influence of this new protocol
on the number of steal requests in the case of unit independent tasks.


We define the potential of the system at time $t$ to be:
\begin{equation*}
  \Phi(t)=\sum_{i=1}^m \Bigl( w_i(t)^\nu-w_i(t) \Bigr).
\end{equation*}
Let us first compute the decrease of the potential when processor $i$
receives $k\ge1$ steal requests. If $w_i(t)>0$, it can be written
$w_i(t)=(k+1)q+b$ with $0\le b<k+1$. We neglect the decrease of potential
due to the execution tasks ($\nu>1$ implies that execution of tasks
decreases the potential). 

After one time step and $k$ steal requests, the work will be divided into $r$
parts with $q+1$ tasks and $k+1-r$ parts with $q$ tasks. $\sum_i w_i(t)$
does not vary during the stealing phase. Therefore, the difference of
potential due to these $k$ work requests is
\begin{equation*}
  \delta^k_i=((k+1)q+b)^\nu - b(q+1)^\nu - (k+1-b)q^\nu.
\end{equation*}
Let us denote $\alpha\bydef b/(k+1)\in[0;1)$ and let $f(x)=(x+\alpha)^\nu +
(1-2^{1-\nu})(x+\alpha)-(1-\alpha)x^\nu-\alpha (x+1)^\nu$. The first
derivative of $f$ is
$f'(x)=\nu(x+\alpha)^{\nu-1}+(1-2^{1-\nu})-\nu(1-\alpha)x^{\nu-1}-\alpha(x+1)^{\nu-1}$
and the derivative of $f'$ is $f''(x)=\nu(1-\nu)( (x+\alpha)^{\nu-2} -
(1-\alpha) x^{\nu-2}-\alpha (x+1)^{\nu-2}$. As $\nu<3$, the function
$x\mapsto x^{\nu-2}$ is concave which implies than $f''(x)\ge0$. Therefore,
$f'$ is increasing. Moreover,
$f'(0)=\nu(\alpha^{\nu-1}-\alpha)+(1-2^{1-\nu})\ge0$.
This shows that for all $x$, $f'(x)\ge0$ and that
$f$ is increasing. The value of $f$ in $0$ is
$f(0)=\alpha^\nu-(1-2^{1-\nu})\alpha-\alpha=\alpha^\nu(1-(2\alpha)^{1-\nu})\ge0$
which implies that for all $x$, $f(x)\ge0$.

Recall that $w_i(t)=(k+1)q+b$ and $\alpha=b/(k+1)$. Using the notation $f$
and the fact that $(k+1)^{1-\nu}\le2^{1-\nu}$, the decrease of potential
$\delta^k_i$ can be written
\begin{align}
  \delta^k_i &= (1-(k+1)^{1-\nu}) \cdot (w_i(t)^\nu-w_i(t))+(k+1) \cdot f(q) \nonumber \\
& \ge(1-(k+1)^{1-\nu}) \cdot (w_i(t)^\nu-w_i(t)). 
  \label{eq:1/k+1}
\end{align}

Let $q_k(r)$ be the probability for a processor to receive $k$ work
requests when $r$ processors are stealing. $q_k(r)$ is equal
to:
\begin{equation*}
  q_k(r)=\lp\begin{array}{c}r\\k\end{array}\rp
  \frac{1}{(m-1)^k}\lp\frac{m-2}{m-1}\rp^{r-k}
\end{equation*} 
The expected decrease of the potential caused by the steals on processor
$i$ is equal to $\sum_{k=0}^{r}\delta^k_iq_k(r)$. Using
equation~\eqref{eq:1/k+1}, we can bound the expected potential at time
$t+1$ by
\begin{align*}
\expect{\Phi_t - \Phi_{t+1} \mid\F_t}
&= \sum_{i=0}^m \sum_{k=0}^{r} \delta^k_i \cdot q_k(r) \\
\expect{\Phi_{t+1}\mid\F_t}
&\le \Bigl( 1-\sum_{k=0}^r (1-(k+1)^{1-\nu})
\cdot q_k(r) \Bigr) \cdot \Phi_t 
\end{align*}

\nomore{
When $\nu=2$, and using Equation~\eqref{eq:1/k+1}, the expected decrease of
potential is bounded by
\begin{eqnarray*}
  \sum_i\sum_{k=0}^{r}\delta^k_iq_k(r) &\ge& 
  \sum_i(w_i(t)^2-w_i(t))\sum_{k=0}^{r}\lp
  1-\frac{1}{k+1}\rp q_k(r) \\
  &=& \sum_i(w_i(t)^2-w_i(t))\lp1-\frac{m-1}{r+1}\lp1-
  \lp\frac{m-2}{m-1}\rp^{r+1}\rp\rp 
\end{eqnarray*}
This shows that $\sexpect{\Phi_{t+1}\mid\F_t}\le h(r_t)\Phi_t$ where
\begin{equation*}
  h(r_t)=\frac{m-1}{r_t+1}\lp1-\lp\frac{m-2}{m-1}\rp^{r_t+1}\rp 
\end{equation*}
Deriving with respect to $r$ shows that $(r)/-\log_2h(r)$ is
increasing. Thus $\lambda\bydef\max_{1\le r\le
  m}r/-m\log_2h(r)=(m-1)/-m\log_2h(1)$. A direct computation shows that
$\lambda\le 1/-\log_2(1-1/e)$. See Appendix~\ref{apx:lambda_coop} for
details. 
}

\begin{theorem}
  The makespan $C^{\mathrm{coop}}_{\max}$ of $W=n$ unit independent tasks
  scheduled with cooperative work stealing satisfies:
  \begin{itemize}
  \item[(i)]
    ~~$\displaystyle \expect{C^{\mathrm{coop}}_{\max}} \le \frac{W}{m} +
    2.88\cdot\log_2 W+3.4$
  \item[(ii)]
    ~~$\displaystyle \proba{C^{\mathrm{coop}}_{\max} \ge \frac{W}{m} +
      2.88\cdot \log_2 W + 2+
      \log_2\lp\frac1\epsilon\rp} \le \epsilon$.
  \end{itemize}
\end{theorem}
\begin{proof}
  The proof is very similar to the one of Theorem~\ref{thm:unit}.
Let
$$h(r)\bydef1 -\sum_{k=0}^{r}(1-(k+1)^{1-\nu})\cdot q_k(r)$$ and
$$\lambda^{\mathrm{coop}}(\nu)\bydef\max_{1\le r\le m}
\frac{r}{-m\cdot\log_2h(r)}.$$
Using Theorem~\ref{thm:main}, we have:
\begin{equation*}
  \expect{C^{\mathrm{coop}}_{\max}} \le \frac{W}{m} +
  \nu\lambda^{\mathrm{coop}}(\nu)\cdot\log_2 W+\frac{\lambda(\nu)}{\ln2}+1. 
\end{equation*}
In the general case the exact computation of
$h(r)$ is intractable. However, by a
numerical computation, one can show that 
$3\lambda^{\mathrm{coop}}(3)<2.88$.
  
When $\Phi_t<1$, we have $\sum_i
  w_i(t)^\nu-w_i(t)<1$. This implies that for all processor $i$, $w_i(t)$
  equals $0$ or $1$.  This adds (at most) one step of computation at the
  end of the schedule. As $\lambda(3)/\ln(2) + 1 + 1 = 3.4$, we obtain
the calimed bound.
  \qed
\end{proof}

\nomore{
In particular, the factor in front of $\log_2 W$ is
$-2/\log_2(1-\frac{1}{e})<3.02$.  A similar computation could have been
done replacing $w_i(t)^2$ by $w_i(t)^\nu$. In that case, one can show that
if a processor $i$ receives $k$ steals requests, then the decrease of
potential is $(1-(k+1)^{1-\nu})(w_i(t)^\nu-w_i(t))$. 
}

Compared to the situation with no cooperation among thieves, the number of
steal requests is reduced by a factor $3.24/2.88\approx 12\%$. We will see
in Section~\ref{sec:simu} that this is close to the value obtained by
simulation.

\begin{remark}
The exact computation can be accomplished for $\nu=2$ \citep{isaac}
and leads to a constant factor
of $2\lambda^{\mathrm{coop}}(2)\le -2/\log_2(1-\frac{1}{e})<3.02$.
\end{remark}

\section{Weighted independent tasks}
\label{sec:weighted}

In this section, we analyze the number of work requests for weighted
independent tasks. Each task $j$ has a processing time $p_j$ which is
unknown. When an idle processor attempts to steal a processor, half of the
tasks of the victim are transfered from the active processor to the idle
one. A task that is currently executed by a processor cannot be stolen. If
the victim has $2k(+1)$ tasks (plus one for the task that is currently
executed), the work is split in $k(+1)$, $k$. If the victim has $2k+1(+1)$
tasks, the work is split in $k(+1)$, $k+1$. 

In all this analysis, we consider that the scheduler does not know the
weight of the different tasks $p_j$. Therefore, when the work is split in
two parts, we do not assume that the work is split fairly (see for example
Figure~\ref{fig:evol_weighted}) but only that the number of tasks is split
in two equal parts.

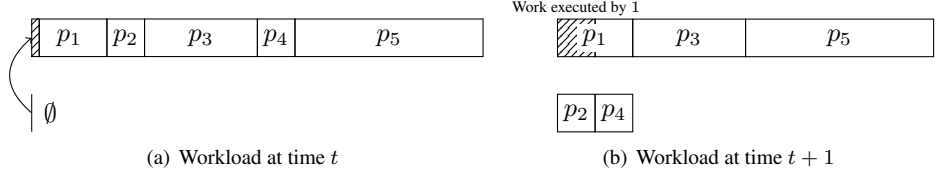
\begin{figure}[ht]
  \centering
    \subfigure[Workload at time $t$]{
    \begin{tikzpicture}[yscale=.5]
      \tikzstyle{executed}=[pattern=north east lines];
      \draw[executed] (0,0) rectangle (.1,1); 
      \draw (0.1,0) rectangle (1,1); \node at (.5,.5) {$p_1$};
      \draw (1,0) rectangle (1.5,1); \node at (1.25,.5) {$p_2$};
      \draw (1.5,0) rectangle (3,1); \node at (2.25,.5) {$p_3$};
      \draw (3,0) rectangle (3.5,1); \node at (3.25,.5) {$p_4$};
      \draw (3.5,0) rectangle (6,1); \node at (4.75,.5) {$p_5$};
      \draw (0,-2) rectangle (0,-1); \node at (.25,-1.5) {$\emptyset$};
      \draw (0,-1.5) edge[bend left,->] (0,.5);
    \end{tikzpicture}
  }~~~
    \subfigure[Workload at time $t+1$]{
    \begin{tikzpicture}[yscale=.5]
      \tikzstyle{executed}=[pattern=north east lines];
      \draw[executed] (0,0) rectangle (.5,1); 
      \node at (.25,1.3) {\tiny Work executed by $1$};
      \draw (.5,0) rectangle (1,1); 
      \node[fill=white,inner sep=2pt] at (.5,.5) {$p_1$};
      \draw (1,0) rectangle (2.5,1); \node at (1.75,.5) {$p_3$};
      \draw (2.5,0) rectangle (5,1); \node at (3.75,.5) {$p_5$};
      \draw (0,-2) rectangle (0.5,-1); \node at (.25,-1.5) {$p_2$};
      \draw (.5,-2) rectangle (1,-1); \node at (.75,-1.5) {$p_4$};
    \end{tikzpicture}
  }
  \caption{Evolution of the repartition of tasks during one time step. At
    time $t$, one processor has all the tasks. $p_1$ can not be stolen
    since the processor $1$ has already started executing it.  After one
    work request done by the second processor, one processor has $3$ tasks
    and one has $2$ tasks but the workload may be very different, depending
    on the processing times $p_j$. }
  \label{fig:evol_weighted}
\end{figure}
\subsection{Definition of the potential function and expected decrease}

As the processing times are unknown, the work cannot be shared evenly
between both processors and can be as bad as one processor getting all the
smallest tasks and one all the biggest tasks (see
Figure~\ref{fig:evol_weighted}).  Let us call $w_i(t)$ the \emph{number of
  tasks} possessed by the processor $i$. The potential of the system at
time $t$ is defined as:
\begin{equation}
  \Phi_t \bydef \sum_{i} \lp w_i(t)^\nu - w_i(t)\rp.
  \label{eq:pot_weighted}
\end{equation}
During a work request, half of the tasks are transfered from an active
processor to the idle processor. If the processor $j$ is stealing tasks
from processor $i$, the number of tasks possessed by $i$ and $j$ at time
$t+1$ are $w_j(t+1)=\ceil{w_i(t)/2}$ and
$w_i(t+1)=\floor{w_i(t)/2}$. Therefore, the decrease of potential is equal
to the one of the cooperative steal of Equation~\ref{eq:1/k+1} for $k=1$:
$$\delta_i\ge(1-2^{1-\nu})\cdot \lp w_i(t)^\nu-w_i(t)\rp.$$
Following the analysis of 
Section~\ref{ssec:coop}, this shows that in average:
\begin{equation}
  \expect{\Phi_{t+1}} \le (1-(1-2^{1-\nu})q(r))\cdot \Phi_t. 
  \label{eq:delta_phi_weighted}
\end{equation}

\subsection{Bound on the makespan}

Equation~\ref{eq:delta_phi_weighted} allows us to apply
Theorem~\ref{thm:main} to derive a bound on the makespan of weighted tasks
by the distributed list scheduling algorithm. This bound differs from the
one for unit tasks only by an additive term of $p_{\max}$.

\begin{theorem}
  \label{thm:weighted}
  Let $p_{\max}\bydef\max p_j$ be the maximum processing times.  The expected
  makespan to schedule $n$ weighted tasks of total processing time $W=\sum
  p_j$ by DLS is bounded by
  \begin{align*}
    \expect{C_{\max}} &\le \frac{W}{m} + \frac{m-1}{m}p_{\max} +
    3.24 \cdot \lp \log_2 n+\frac{1}{2\ln2}
    \rp+1
  \end{align*}
\end{theorem}

\begin{proof}
  Let $\Phi_t$ be the potential defined by
  Equation~\ref{eq:pot_weighted}. At time $t=0$, the potential of the
  system is bounded by $W^{\nu}-W$. Therefore, by Theorem~\ref{thm:main}, the
  number of work requests before $\Phi_t<1$ is bounded by
  \begin{equation*}
    m\cdot\lambda\cdot\lp\log_2\Phi_0+1+\frac1{\ln2}\rp \le
    m\cdot\nu\lambda(\nu)\cdot\lp2\log_2W+1+\frac1{\ln2}\rp,
  \end{equation*}
  where $\nu\lambda(\nu)<3.24$ is the same constant as the bound for the unit tasks
  with the potential function $\sum_iw_i^\nu$ of Theorem~\ref{thm:unit_nu}.

  As $\Phi_t\in\N$, $\Phi_t<1$ implies that $\Phi_t=0$. Moreover, by
  definition of $\Phi_t$, this implies that for all $i$:
  $w_i(t)^\nu-w_i(t)=0$, which implies that for all $i$:
  $w_i(t)\le1$. Therefore, once $\Phi_t$ is equal to $0$, there is at most
  one task per processor. This phase can last for at most $p_{\max}$ unit
  of time, generating at most $(m-1)p_{\max}$ work requests. 
  \qed
\end{proof}

\begin{remark}
The same analysis applies for the cooperative stealing scheme
 of Section~\ref{ssec:coop}
leading to the same improved bound in $2.88 \log_2 n$ instead of
$3.24 \log_2 n$.
\end{remark}

\section{Tasks with precedences}
\label{sec:dag}
In this section, we show how the well known non-blocking work stealing of \cite{ABP}
(denoted ABP in the sequel) can be analyzed with our method which provides tighter bounds
for the makespan. We first recall the WS scheduler of ABP, then we show how to define
the amount of work on a processor $w_i(t)$, finally we apply the analysis of Section~\ref{sec:thm}
to bound the makespan.

\subsection{ABP work-stealing scheduler}

Following~\cite{ABP}, a multithreaded computation is modeled as a directed acyclic
graph $G$ with $W$ unit tasks task and edges define precedence constraints. 
There is a single source task and the out-degree is at most $2$. 
The critical path of $G$ is denoted by $D$.
ABP schedules the DAG $G$ as follows.
Each processor $i$ maintains a double-ended queue (called a deque) $Q_i$ of ready tasks.
At each slot, an active processor $i$ with a non-empty deque executes the task at
the bottom of its deque $Q_i$; once its execution is completed, this task is popped
from the bottom of the deque,  enabling --~\textit{i.e.} making ready~-- $0$, $1$ or $2$ child
tasks that are pushed at the bottom of $Q_i$. At each top, an idle processor $j$ with
an empty deque $Q_j$ becomes a thief: it performs a steal request on another randomly
chosen victim deque; if the victim deque contains ready tasks, then its top-most task is
popped and pushed into the deque of one of its concurrent thieves. 
If $j$ becomes active just after its steal request, the steal request is said
successful. Otherwise, $Q_j$ remains empty and the steal request fails which may occur
in the three following situations: either the victim deque $Q_i$ is empty; or, $Q_i$ contains only
one task currently in execution on $i$; or, due to contention, another thief performs
a successful steal request on $i$ simultaneously.

\subsection{Definition of $w_i(t)$}

Let us first recall the definition of the {\em enabling tree} of~\cite{ABP}. 
If the execution of task $u$ enables task $v$, then the edge $(u,v)$ of $G$ is an enabling edge.
The sub-graph of $G$ consisting of only enabling edges forms a rooted tree called the
enabling tree. 
We denote by $h(u)$ the height of a task $u$ in the enabling tree. The root of the DAG
has height $D$. Moreover, it has been shown in \cite{ABP} that tasks in the deque
have strictly decreasing height from top to bottom except for the two bottom most tasks
which can have equal heights.

We now define $w_i(t)$, the amount of work on processor $i$ at time $t$.
Let $h_t$ be the maximum height of all tasks in the deque.  If the deque contains
at least two tasks including the one currently executing we define $w_i(t)=(2\sqrt{2})^{h_t}$.
If the deque contains only one task currently executing we define $w_i(t)=\frac12\cdot(2\sqrt{2})^{h_t}$.
The following lemma states that this definition of $w_i(t)$ behaves in a similar
way than the one used for the independent unit tasks analysis of Section~\ref{sec:unit}.

\begin{lemma} \label{lem:ABPhalved}
For any active processor $i$, we have $w_i(t+1) \leq w_i(t)$. 
Moreover, after any successful steal request from a processor $j$ on $i$,
$w_i(t+1) \leq w_i(t)/2$ and $w_j(t+1) \leq w_i(t)/2$
and if all steal requests are unsuccessful we have
$w_i(t+1) \leq w_i(t)/\sqrt{2}$.
\end{lemma}

\begin{proof}
We first analyze the execution of one task $u$ at the bottom
of the deque. Executing task $u$ enables at most two tasks and
these tasks are the children of $u$ in the enabling tree. 
If the deque contains more than one task, the top most task has
height $h_t$ and this task is still in the deque at time $t+1$.
Thus the maximum height does not change and $w_i(t)= w_i(t+1)$.
If the deque contains only one task, we have
$w_i(t) = \frac12\cdot(2\sqrt{2})^{h_t}$ and
$w_i(t+1) \leq (2\sqrt{2})^{h_t-1}$. Thus $w_i(t+1)\leq w_i(t)$.

We now analyze a successful steal from processor $j$. In this case, the deque of processor
$i$ contains at least two tasks and  $w_i(t) = (2\sqrt{2})^{h_t}$. The stolen task is one
with the maximum height and is the only task in the deque of processor $j$
thus $w_j(t+1) =\frac12\cdot(2\sqrt{2})^{h_t} \leq w_i(t)/2$. For the processor $i$, either
its deque contains only one task after the steal with height at most $h_t$
and $w_i(t+1)\leq \frac12\cdot(2\sqrt{2})^{h_t} \leq w_i(t)/2$, either there are still more
than $2$ tasks and $w_i(t+1)\leq (2\sqrt{2})^{h_t-1} < w_i(t)/2$.

Finally, if all steal requests are unsuccessful, the deque of processor $i$ contains
at most one task. If the deque is empty $w_i(t+1)=w_i(t)=0$ and thus 
$w_i(t+1) \leq w_i(t)/\sqrt{2}$. If the deque contains exactly one task, 
$w_i(t) = \frac12\cdot(2\sqrt{2})^{h_t}$ and $w_i(t+1) \leq (2\sqrt{2})^{h_t-1}$
thus $w_i(t+1)\leq w_i(t)/\sqrt{2}$. \qed
\end{proof}

\subsection{Bound on the makespan}

To study the number of steals, we follow the analysis presented in Section~\ref{sec:thm}
with the potential function $\Phi(t) = \sum_i w_i(t)^2$.
Using results from lemma~\ref{lem:ABPhalved}, we compute the decrease
of the potential $\delta_i(t)$ due to steal requests on processor $i$ by distinguishing 
two cases. If there is a successful steal from processor $j$,
$$ \delta_i(t) = w_i(t)^2 - w_i(t+1)^2 - w_j(t+1)^2
\geq w_i(t)^2 - 2\cdot\Bigl(\frac{w_i(t)}{2}\Bigr)^2
\geq \frac12 \cdot w_i(t)^2.$$
If all steals are unsuccessful, the decrease of the potential is
$$ \delta_i(t) = w_i(t)^2 - w_i(t+1)^2
\geq w_i(t)^2 - \Bigl(\frac{w_i(t)}{\sqrt{2}}\Bigr)^2
\geq \frac12 \cdot w_i(t)^2. $$
In all cases, $\delta_i(t) \geq w_i(t)^2 / 2$. We obtain the expected potential
at time $t+1$ by summing the expected decrease on each active processor:
\begin{align*}
\expect{\Phi_t - \Phi_{t+1} \mid \F_t } &\geq \sum_{i=0}^m \frac{w_i(t)^2}{2} \cdot q(r_t) \\
\expect{\Phi_{t+1} \mid \F_t } &\leq \Bigl(1-\frac{q(r_t)}{2}\Bigr) \cdot \Phi(t)
\end{align*}
Finally, we can state the following theorem.
\begin{theorem} \label{thm:ABPWS}
On a DAG composed of $W$ unit tasks, with critical path $D$, one source and out-degree at most $2$,
 the makespan of ABP work stealing verifies:
\begin{itemize}
\item[(i)] ~~~~
$\displaystyle \expect{C_{\max}} \leq \frac{W}{m} + \frac{3}{1-\log_2(1+\frac1e)} \cdot D+1
< \frac{W}{m} + 5.5 \cdot D+1.$
\item[(ii)] ~~~~
$\displaystyle \proba{C_{\max} \geq \frac{W}{m} + \frac{3}{1-\log_2(1+\frac1e)} \cdot
\Bigl( D+\log_2\frac{1}{\epsilon} \Bigr)+1 } 
\leq \epsilon $
\end{itemize}
\end{theorem}
\begin{proof}
The proof is a direct application of Theorem~\ref{thm:main}. As in the initial step there is only
one non empty deque containing the root task with height $D$, the initial potential is
$$\Phi(0) =\Bigl(\frac12\cdot \Bigl(2\sqrt{2}\Bigr)^D\Bigr)^2.$$
Thus the expected number of steal requests before $\Phi(t) < 1$ is bounded by
\begin{align*}
\expect{R} &\leq \lambda \cdot m \cdot \log_2 \Bigl[\Bigl(\frac12\cdot \Bigl(2\sqrt{2}\Bigr)^D\Bigr)^2\Bigr]
+ m\cdot \Bigl(1+\frac{\lambda}{\ln(2)}\Bigr) \\
&\leq 2\lambda \cdot m \cdot D \cdot \log_2 (2\sqrt{2}) + m\cdot\Bigl( 1 + \frac{\lambda}{\ln(2)} - 2\lambda\Bigr)  \\
& \leq 3\lambda \cdot m \cdot D \hspace{3cm} (\text{as}~~1+\lambda/\ln(2)-2\lambda<0)
\end{align*}
where $\lambda=(1-\log_2(1+1/e))^{-1}$ is the same constant as the bound
for the unit tasks of Section~\ref{sec:unit}.

Moreover, when $\Phi(t) < 1$,
we have $\forall i, w_i(t) < 1$. There is at most one task of height $0$
in each deque, \textit{i.e.} a leaf of the enabling tree which cannot enable any other task.
This last step generates at most $m-1$ additional steal requests.
In total, the expected number of steal
requests is bounded by  $\expect{R} \leq 3\lambda \cdot m \cdot D+m-1$. The bound on the
makespan is obtained using the relation $m\cdot C_{\max} = W+R$.

The proof of (i) applies \textit{mutatis mutandis} to prove the bound in probability (ii).
\qed
\end{proof}

\paragraph{Remark.} 
In \cite{ABP}, the authors established the upper bounds :
$$
\expect{C_{\max}} \leq \frac{W}{m} +32 \cdot D ~~~~\text{and}~~~~
\proba{C_{\max} \geq \frac{W}{m} + 64\cdot D+16\cdot\log_2\frac{1}{\epsilon} }  \leq \epsilon
$$
in Section 4.3, proof of Theorem 9.
Our bounds greatly improve the constant factors of this previous result.

\section{Experimental study}
\label{sec:simu}

The theoretical analysis gives an upper bounds on the expected value of
the makespan and deviation from the mean for the various models we considered.
In this section,
we study experimentally the distribution of the makespan. Statistical tests
give evidence that the makespan for independent tasks follows a generalized
extreme value (gev) distribution~\citep{gev}. This was expected since such a distribution
arises when dealing with maximum of random variables. For tasks with dependencies,
it depends on the structure of the graph: DAGs with short critical path still follow
a gev distribution but when the critical path grows, it tends to a gaussian distribution.
We also study in more details the overhead to $W/m$ and show that it is approximately
$2.37\log_2 W$ for unit independent tasks which is close to the theoretical result
of $3.24\log_2 W$ (\textit{cf.} Section~\ref{sec:further}).

We developed a simulator that strictly follows our model.
At the beginning, all the tasks are given to processor $0$ in order to be
in the worst case, \textit{i.e.} when the initial potential $\Phi_0$ is maximum.
Each pair ($m$,$W$) is simulated $10000$ to get accurate results, with a coefficient of
variation about $2\%$.

\subsection{Distribution of the makespan}
\label{ssec:distrib}


\begin{figure}[bt]
\centering
\subfigure[Unit Tasks]{
	\includegraphics[width=0.23\textwidth]{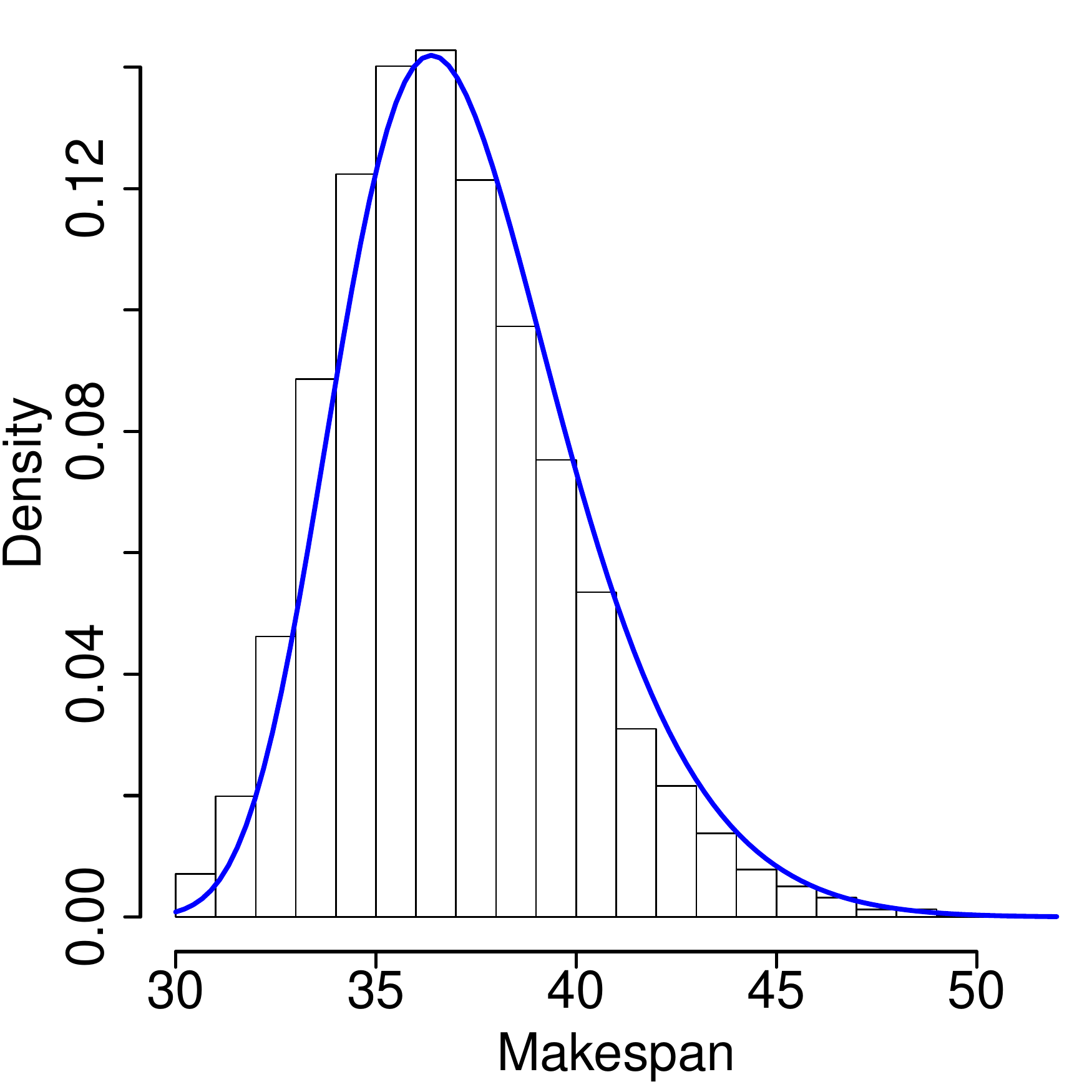}
	\label{fig:distrib-unit}
}%
\subfigure[Weighted Tasks]{
	\includegraphics[width=0.23\textwidth]{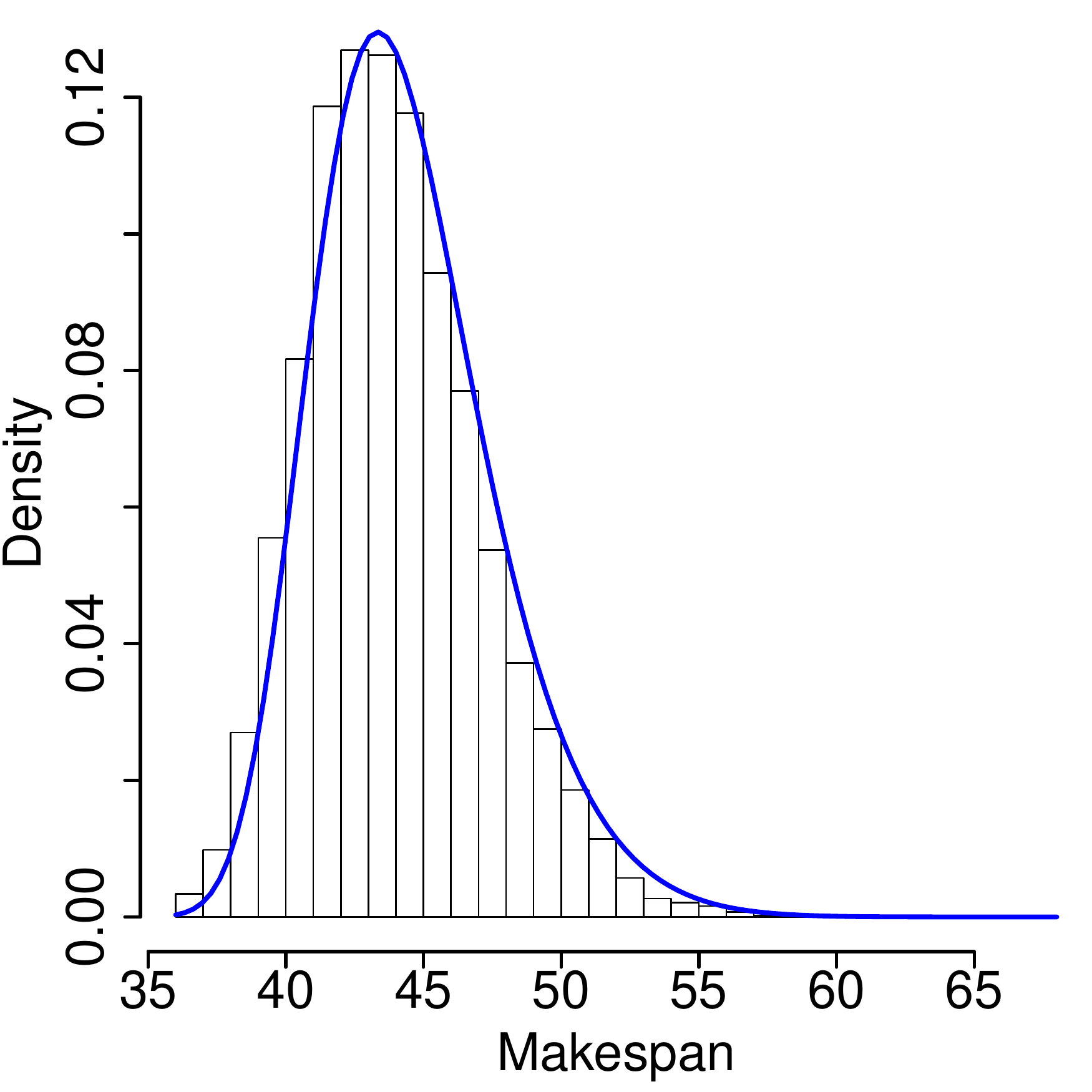}
	\label{fig:distrib-weighted}
}
\subfigure[DAG (short $D$)]{
	\includegraphics[width=0.23\textwidth]{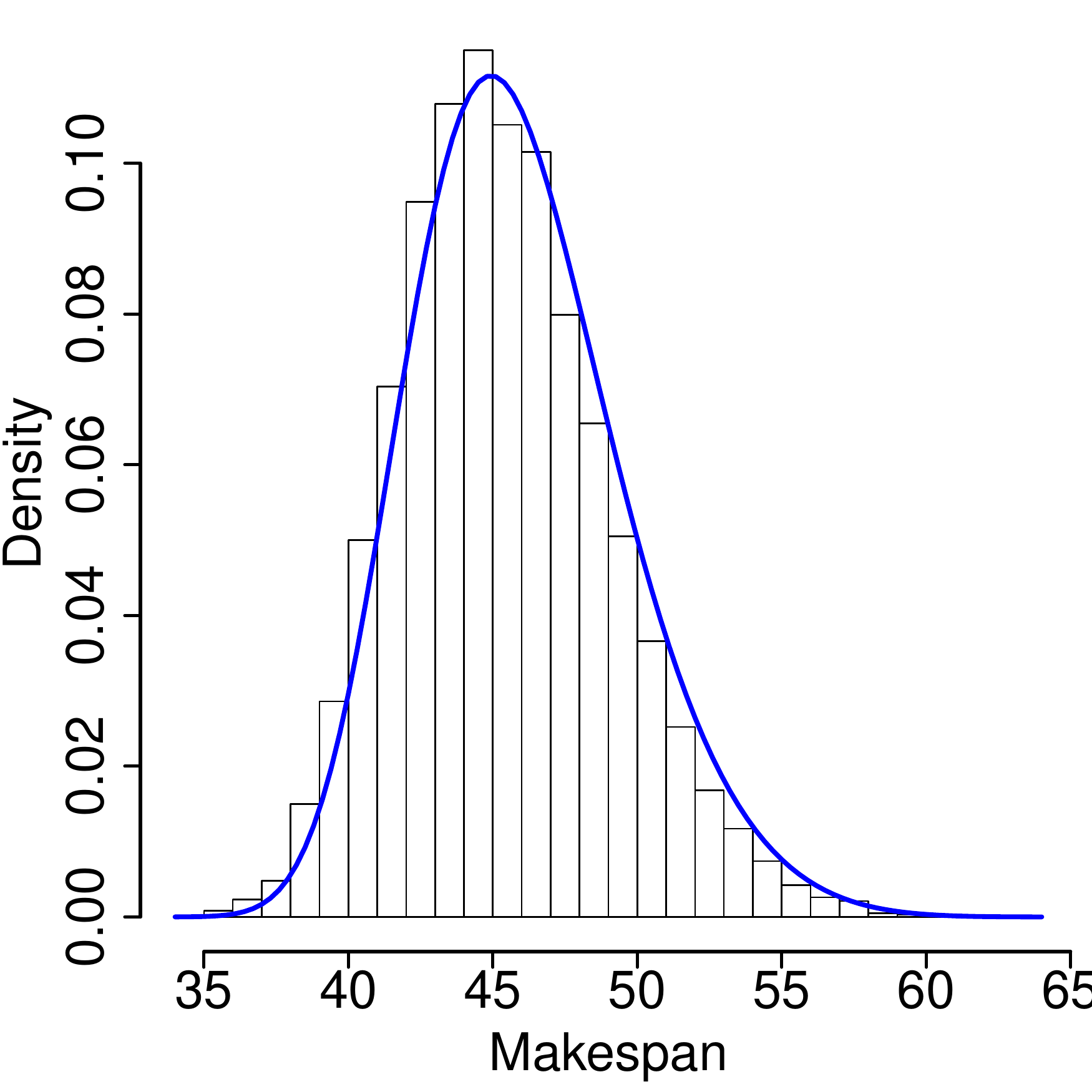}
	\label{fig:distrib-dag-short}
}%
\subfigure[DAG (long $D$)]{
	\includegraphics[width=0.23\textwidth]{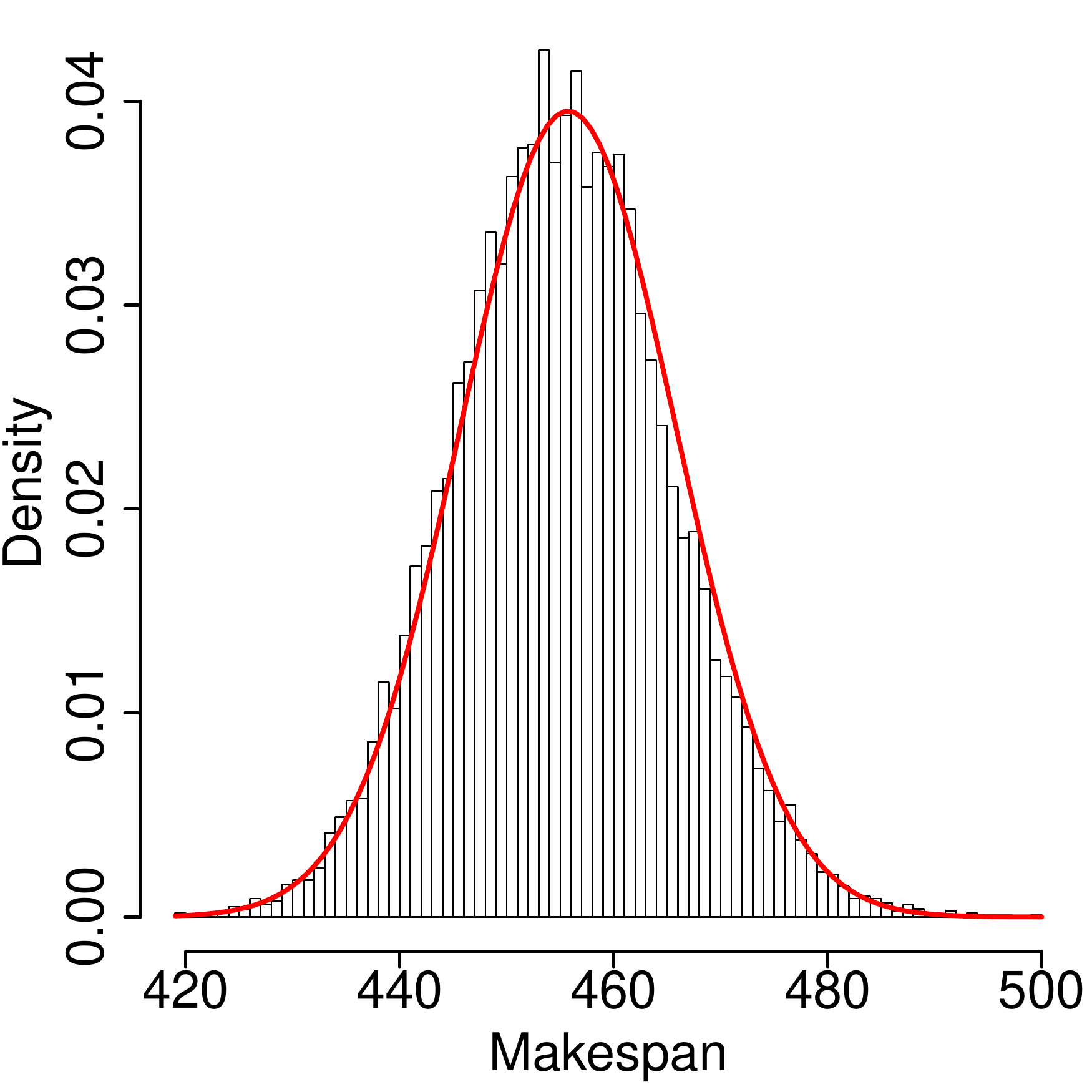}
	\label{fig:distrib-dag-long}
}%
\caption{%
\label{fig:distrib}%
Distribution of the makespan for unit independent tasks~\ref{fig:distrib-unit},
weighted independent tasks~\ref{fig:distrib-weighted} and tasks with
dependencies~\ref{fig:distrib-dag-short} and~\ref{fig:distrib-dag-long}.
The first three models follow a gev distribution (blue curves), the last
one is gaussian (red curve).}
\end{figure}

We consider here a fixed workload $W=2^{17}$ on $m=2^{10}$ processors for independent tasks
and $m=2^7$ processors for tasks with dependencies. For the weighted model,
processing times were generated randomly and uniformly between $1$ and $10$. For 
the DAG model, graphs have been generated using a layer by layer method. We generated
two types of DAGs, one with a short critical path (close to the minimum possible $\log_2 W$) and
the other one with a long critical path (around $W/4m$ in order to keep enough
tasks per processor per layer).
Fig.~\ref{fig:distrib} presents histograms for $C_{\max}-\lceil W/m\rceil$.

The distributions of the first three models (a,b,c in Fig.~\ref{fig:distrib})
are clearly not gaussian: they are asymmetrical with an heavier right tail. To fit
these three models, we use the generalized extreme value (gev) distribution~\citep{gev}. In the
same way as the normal distribution arises when studying the sum of independent and identically
distributed (iid) random variables, the gev distribution arises when studying the maximum of
iid random variables. The extreme value theorem, an equivalent of the central limit theorem for maxima,
states that the maximum of iid random variables converges in distribution to a gev distribution.
In our setting, the random variables measuring the load of each processor
are not independent, thus the extreme value theorem cannot apply directly. However, it is possible
to fit the distribution of the makespan to a gev distribution. In Fig.~\ref{fig:distrib}, the fitted
distributions (blue curve) closely follow the histograms. To confirm this graphical approach, we
performed a goodness of fit test. The $\chi^2$ test is well-suited to our data because the distribution
of the makespan is discrete.
We compared the results of the best fitted gev to the best fitted gaussian. The $\chi^2$ test strongly rejects
the gaussian hypothesis but does not reject the gev hypothesis with a p-value of more than $0.5$.
This confirms that the makespan follows a gev distribution.
We fitted the last model, DAG with long critical path, with a gaussian
(red curve in Fig.~\ref{fig:distrib-dag-long}).
In this last case, the completion time of each layer of the DAG should correspond to a gev distribution
but the total makespan, the sums of all layers, should tend to a gaussian by the central limit theorem.
Indeed the $\chi^2$ test does not reject the gaussian hypothesis with a p-value around $0.3$.

\subsection{Study of the $\log_2 W$ term}

\begin{figure}[bt]
\centering
\subfigure{
\includegraphics{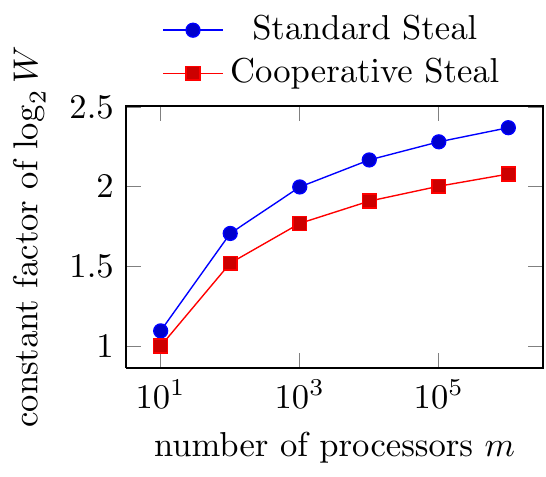}}
\subfigure{
\includegraphics{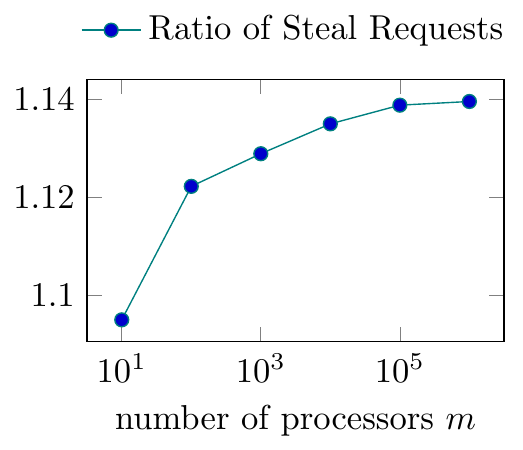}
}
\caption{(Left) Constant factor of $\log_2 W$ against the number of
processors for the standard steal and the cooperative steal.
(Right) Ratio of steal requests (standard$/$cooperative).
\label{fig:slope}%
}
\end{figure}

We focus now on unit independent tasks
as the other models rely on too many parameters (the choice of the processing
times for weighted tasks and the structure of the DAG for tasks with dependencies).
We want to show that the number of work requests is proportional to $\log_2 W$ and
study the proportionality constant. We first launch simulations
with a fixed number of processors $m$ and a wide range of work in successive powers of
$10$. A linear regression confirms the linear dependency in $\log_2 W$ with a
coefficient of determination
("r squared") greater than $0.9999$\footnote{the closer to 1, the better}.

Then, we obtain the slope of the regression for various number of processors.
The value of the slope tends to a limit
around $2.37$~(\textit{cf.} Fig.~\ref{fig:slope}(left)).
This shows that
the theoretical analysis of Theorem~\ref{thm:unit} is almost accurate with a
constant of approximately $3.24$.
We also study the constant factor of $\log_2 W$ for the cooperative steal
of Section~\ref{sec:further}. The theoretical value of $2.88$ is again close to the
value obtained by simulation $2.08$ (\textit{cf.} Figure~\ref{fig:slope}(left)).
The difference between the theoretical and the practical values can be explained
by the worst case analysis on the number of steal requests per time step
in Theorem~\ref{thm:main}.

Moreover, simulations in Fig.~\ref{fig:slope}(right) show that the ratio
of steal requests between standard and cooperative steals goes
asymptotically to 14\%. The ratio between the two corresponding theoretical
bounds is about 12\%. This indicates that the biais introduced by our
analysis is systematic and thus, our analysis may be used as a good
prediction while using cooperation among thieves. 


\section{Concluding Remarks}
\label{sec:conclusion}

In this paper, we presented a complete analysis of the cost of distribution
in list scheduling.  We proposed a new framework, based on potential
functions, for analyzing the complexity of distributed list scheduling
algorithms. In all variants of the problem, we succeeded to characterize
precisely the overhead due to the decentralization of the list. These
results are summarized in the following table comparing makespans for
standard (centralized) and decentralized list scheduling.

\renewcommand{\arraystretch}{2.1}
\hspace{-.5cm}\begin{tabular}{l|c|c}
  &Centralized& Decentralized (WS)\\
  \hline
  Unit Tasks ~($W=n$)& 
  $\displaystyle\ceil{\frac{W}m}$& $\displaystyle\frac{W}m+3.24\log_2
  W+3.33$\\
  
  ~~~~~~~Initial repartition & -- &  $\displaystyle\frac{W}m+1.83\log_2
  \sum_{i=0}^m\Bigl(w_i-\frac{W}{m}\Bigr)^2+3.63$ \\ 

  ~~~~~~~Cooperative
  & -- &$\displaystyle\frac{W}m+2.88\log_2 W+3.4$\\
  
  \hline
  Weighted Tasks&  $\displaystyle\frac{W}m+\frac{m-1}m\cdot p_{\max}$ &
  $\displaystyle\frac{W}{m} + \frac{m-1}{m}\cdot p_{\max} +
    3.24 \log_2 n+3.33$ \\
  \hline
  Tasks w. predecences & $\displaystyle \frac{W}{m} + \frac{m-1}{m} \cdot D
  $ & $\displaystyle\frac{W}{m} + 5.5 D+1$ \\
  \hline
\end{tabular} \medskip

In particular, in the case of independent tasks, the overhead due to the
distribution is small and only depends on the number of tasks and not on
their weights. In addition, this analysis improves the bounds for
the classical work stealing algorithm of \cite{ABP} from $32 D$ to $5.5 D$.  We
believe that this work should help to clarify the links between classical
list scheduling and work stealing.

Furthermore, the framework to analyze DLS algorithms described in this paper is more
general than the method of \cite{ABP}. Indeed, we do not assume a specific
rule (\textit{e.g.} depth first execution of tasks) to manage the local lists. 
Moreover, we do not refer to the structure of the DAG (\emph{e.g.} the depth of a task
in the enabling tree) but on the work contained in each list.
Thus, we plan to extend this analysis to the case of general precedence graphs.

\nomore{
La methodo presente ici est + generique que le WS classique:
\begin{itemize}
\item Pas de regle de gestion locale de liste
\item Pas de notion explicite de structure de graphe.
\end{itemize}
A ce titre, on pense que cela devrait s'etendre a des graphe plus
generaux. 
\marginpar{Marc: traduction?}
}

\begin{acknowledgements}
The authors would like to thank Julien Bernard and Jean-Louis Roch for
fruitful discussions on the preliminary version of this work.
\end{acknowledgements}

\bibliographystyle{spbasic}
\bibliography{paper}

\end{document}